\documentclass[11pt,a4paper,UKenglish]{article}
\pdfoutput=1

\usepackage[margin=1in]{geometry}
\usepackage{amsmath,amsthm,amssymb,booktabs,titlesec,titletoc,enumitem,graphicx,latexsym,url,xcolor,doi}
\usepackage[numbers,sort&compress]{natbib}
\usepackage{thmtools}
\usepackage{microtype,hyperref}
\usepackage[T1]{fontenc}
\usepackage[UKenglish]{babel}
\usepackage[numbib,nottoc]{tocbibind}

\newcommand{\namedref}[2]{\hyperref[#2]{#1~\ref*{#2}}}
\newcommand{\sectionref}[1]{\namedref{Section}{#1}}
\newcommand{\figureref}[1]{\namedref{Figure}{#1}}

\newcommand{\equationref}[1]{\hyperref[#1]{Eq~(\ref*{#1})}}
\newcommand{\theoremref}[1]{\hyperref[#1]{Theorem~\ref*{#1}}}
\newcommand{\lemmaref}[1]{\hyperref[#1]{Lemma~\ref*{#1}}}
\newcommand{\remarkref}[1]{\hyperref[#1]{Remark~\ref*{#1}}}
\newcommand{\definitionref}[1]{\hyperref[#1]{Definition~\ref*{#1}}}
\newcommand{\appendixref}[1]{\hyperref[#1]{Appendix~\ref*{#1}}}

\newtheorem{theorem}{Theorem}
\newtheorem{lemma}{Lemma}
\newtheorem{corollary}{Corollary}
\newtheorem{definition}{Definition}
\theoremstyle{remark}
\newtheorem{remark}{Remark}

\renewcommand{\vec}[1]{\mathbf{#1}}

\DeclareMathOperator*{\ecc}{ecc}

\DeclareMathOperator*{\dist}{dist}
\DeclareMathOperator*{\cent}{center}
\DeclareMathOperator*{\ex}{ex}

\newcommand{\emptysym}{\bot}

\hypersetup{
    colorlinks,
    linkcolor={black},
    citecolor={blue!50!black},
    urlcolor={blue!80!black}
}

\newenvironment{mycover}
               {\list{}{\listparindent 0pt
                        \itemindent    \listparindent
                        \leftmargin    0pt
                        \rightmargin   0pt
                        \parsep        0pt}%
                \raggedright
                \item\relax}
               {\endlist}

\begin{document}

\pagestyle{plain}

\begin{mycover}
    {\LARGE \textbf{Byzantine Approximate Agreement on Graphs}}

\bigskip
\bigskip

\medskip
\textbf{Thomas Nowak}\, $\cdot$\, \href{mailto:thomas.nowak@lri.fr}{\nolinkurl{thomas.nowak@lri.fr}} \\ Universit\'{e} Paris-Sud \& CNRS

\bigskip
\textbf{Joel Rybicki}\, $\cdot$\, \href{mailto:joel.rybicki@ist.ac.at}{\nolinkurl{joel.rybicki@ist.ac.at}} \\ Institute of Science and Technology Austria
\bigskip
\end{mycover}

\paragraph{Abstract.}
Consider a distributed system with $n$ processors out of which $f$ can be Byzantine faulty. In the approximate agreement task, each processor $i$ receives an input value $x_i$ and has to decide on an output value $y_i$ such that
    \begin{enumerate}[noitemsep]
        \item the output values are in the convex hull of the non-faulty processors' input values,
        \item the output values are within distance $d$ of each other. 
    \end{enumerate}
Classically, the values are assumed to be from an $m$-dimensional Euclidean space, where $m \ge 1$.

In this work, we study the task in a discrete setting, where input values with some structure expressible as a graph. Namely, the input values are vertices of a finite graph $G$ and the goal is to output vertices that are within distance $d$ of each other in $G$, but still remain in the graph-induced convex hull of the input values. For $d=0$, the task reduces to consensus and cannot be solved with a deterministic algorithm in an asynchronous system even with a single crash fault. For any $d \ge 1$, we show that the task is solvable in asynchronous systems when $G$ is chordal and $n > (\omega+1)f$, where $\omega$ is the clique number of~$G$. In addition, we give the first Byzantine-tolerant algorithm for a variant of lattice agreement.
For synchronous systems, we show tight resilience bounds for the exact variants of these and related tasks over a large class of combinatorial structures.


\newpage

\section{Introduction}

In a distributed system, processors often need to coordinate their actions by jointly making consistent decisions or collectively agreeing on some data.
While distributed systems can be resilient to failures, the extent to which they do so varies dramatically depending on the underlying communication and timing model, the fault model, and the level of coordination required by the task at hand. Exploring this interplay is at the core of distributed computing.

In this work, we investigate to which degree agreement can be reached in message-passing systems with Byzantine faults when (1) the set of input values has some discrete, combinatorial structure and (2) the set of output values must satisfy some structural closure property over the input values. We consider deterministic algorithms and assume a system with fully-connected point-to-point communication topology consisting of $n$ processors out of which $f$ may experience Byzantine failures, where the faulty processors may arbitrarily deviate from the protocol (e.g., crash, omit messages, or send malicious misinformation).
We consider both asynchronous and synchronous systems. In the former, the processors do not have access to a shared global clock and sent messages may take arbitrarily long (but finite) time to be delivered. In the synchronous case, computation and communication proceeds in a lock-step fashion over discrete rounds.

\subsection{Fault-tolerant distributed agreement tasks} 

Let $P$ denote the set of $n$ processors and $F \subseteq P$ some (unknown) set of faulty processors, where $|F| \le f$. Many distributed agreement problems take the following form: Each processor $i \in P$ receives some input value $x_i \in V$, where $V$ is the set of possible input values. The task is to have every non-faulty processor $i \in P \setminus F$ (irreversibly) decide on an output value $y_i \in V$ subject to some \emph{agreement} and \emph{validity} constraints. These constraints are commonly defined over the sets $X = \{ x_i : i \in P \setminus F\}$ of input  and $Y = \{ y_i : i \in P \setminus F\}$ of output values of \emph{non-faulty} processors. By choosing different constraints, one obtains different types of agreement problems. 

\subsubsection{Consensus and $k$-set agreement} Consensus is one of the most elementary problems in distributed computing~\cite{pease80reaching}: all non-faulty processors should output a single value (agreement) that was the input of some non-faulty processor (validity). A natural generalisation of consensus is the $k$-set agreement problem~\cite{chaudhuri1993more}, which is defined by the following constraints:
\begin{itemize}
    \item agreement: $|Y| \le k$ (all non-faulty processors decide on at most $k$ values),
    \item validity: $Y \subseteq X$ (each decided value was an input of some non-faulty processor).
\end{itemize}
The special case $k=1$ is the consensus problem and is known to be impossible to solve in an asynchronous setting even with $V = \{0,1\}$ under a single crash fault using deterministic algorithms~\cite{fischer85impossibility}. Analogously, $k$-set agreement cannot in general be solved in an asynchronous message-passing systems if there are $f \ge k$ crash faults~\cite{herlihy1998unifying,biely2011easy}. Note that for $k$-set agreement, it is natural to consider also other validity constraints~\cite{deprisco2001kset}.

\subsubsection{Approximate agreement} While consensus and $k$-set agreement cannot in general be solved in an asynchronous system, it is however possible to obtain \emph{approximate} agreement -- in the sense that output values are close to each other -- 
even in the presence of Byzantine faults. Formally, in the (multidimensional) approximate agreement problem, we are given $\varepsilon>0$ and the set $V = \mathbb{R}^m$ of values forms an $m$-dimensional Euclidean space for some $m \ge 1$. The task is to satisfy
\begin{itemize}
    \item agreement: $\dist(y,y') \le \varepsilon$ for any $y,y' \in Y$ (output values are within Euclidean distance~$\varepsilon$),
    \item validity: the set $Y$ is contained in the convex hull $\langle X \rangle$ of the set $X$ of nonfaulty input values.
\end{itemize}
For an arbitrary $m \ge 1$, Mendes et al.~\cite{MHVG15} showed that under Byzantine faults the problem is solvable in asynchronous systems if and only if $n > (m+2)f$ holds.

\subsubsection{Lattice agreement} Lattice agreement is another well-studied relaxation of consensus with applications in renaming problems and obtaining atomic snapshots~\cite{attiya1995atomic,attiya2001adaptive,faleiro2012generalized,zheng2018lattice}. In this problem, the set $V$ of values forms a \emph{semilattice } $\mathbb{L} = (V, \oplus)$, i.e., an idempotent commutative semigroup. The $\oplus$ operator defines a partial order $\le$ over $V$ defined as $u \le v \iff u \oplus v = v$. 
The task 
is to decide on values that lie on a non-trivial chain, i.e., values that are comparable under $\le$:
\begin{itemize}
    \item agreement: $y \le y'$ or $y' \le y$ for any $y,y' \in Y$, 
    \item validity: for any $y \in Y$ there exists some $x \in X$ such that $x \le y$ and $y \le \bigoplus X$.
\end{itemize}
Note that under crash faults the validity condition is usually given as $x_i \le y_i \le \bigoplus \{ x_j : j \in P \}$ for $i \in P \setminus F$, which is less suitable in the context of Byzantine faults since otherwise output values could exit the convex hull defined by the correct processes' input values.

\subsection{Structured agreement problems}

Unlike the $k$-set agreement problem, the approximate and lattice agreement problems impose additional structure on the set $V$ of values. In the former, the values form a (continuous) $m$-dimensional Euclidean space, whereas in the latter there is algebraic structure. Furthermore, the validity conditions require that the output respects some \emph{closure property} on the input values.
In approximate agreement, the closure is given by the convex hull operator in Euclidean spaces, whereas in lattice agreement, the output must reside in the minimal superset of $X$ closed under~$\oplus$. 

Such closure systems have been studied under the notion of \emph{abstract convexity spaces} and have a rich theory~\cite{kay1971axiomatic,Edelman1985,duchet1987convexity,Edelman1988dimension,van1993theory}. A convexity space on $V$ is a collection $\mathcal{C}$ of subsets of $V$ that satisfies
\begin{enumerate}
    \item $\emptyset, V \in \mathcal{C}$,
    \item $A,B \in \mathcal{C}$ implies $A \cap B \in \mathcal{C}$.
\end{enumerate}
As the name suggests, the sets in $\mathcal{C}$ are called \emph{convex} and every convexity space has the natural closure operator, which maps any set $A \in V$ to a minimal convex superset $A \subseteq \langle A \rangle \in \mathcal{C}$ called the convex hull of $A$. \emph{Convex geometries}~\cite{Edelman1985} are an important class of convexity spaces, which satisfy the Minkowski-Krein-Milman property: the closure $\langle A \rangle$ of any set $A \subseteq V$ is the closure of its \emph{extreme points}, where $a \in A$ is an extreme point of $A$ if $a \notin \langle A \setminus a \rangle$. Convex geometries have been studied extensively in a wide variety of \emph{combinatorial} structures, such as graphs and hypergraphs~\cite{Jamison1984Helly,Farber1986Convexity,Farber1987Local,Duchet1988Convex,Pelayo2013,nielsen2009steiner,dourado2013caratheodory}, and partially ordered sets~\cite{duchet1987convexity,Edelman1985,poncet2014semilattice}.

There has been extensive research on developing theory of convexity over \emph{combinatorial} structures, such as graphs and hypergraphs~\cite{Jamison1984Helly,Farber1986Convexity,Farber1987Local,Duchet1988Convex,Pelayo2013,nielsen2009steiner,dourado2013caratheodory}, partially ordered sets~\cite{duchet1987convexity,Edelman1985,poncet2014semilattice}, and so on. Much of the research has focused on identifying analogues to classical convexity invariants, such as Helly, Carath\'eodory, and Radon numbers, in various abstract convexity spaces~\cite{kay1971axiomatic,Jamison1984Helly,Duchet1988Convex,eckhoff1993helly,barbosa2012caratheodory,dourado2013caratheodory}. Convex geometries also have deep connections with matroid and antimatroid theory~\cite{dietrich1989matroids,korte2012greedoids}: convex geometries are duals of antimatroids, and a special class of greedoids, which provide a structural framework for characterising greedy algorithms~\cite{korte1984greedoids,korte2012greedoids}, are convex geometries~\cite{Edelman1985}. Lov\'asz and Saks~\cite{lovasz1993communication} used theory of convex geometries to analyze a broad class of two-party communication complexity problems.

\subsection{Approximate agreement on graphs}

As our main example of an agreement problem with discrete, combinatorial structure, we focus on a problem where the set $V$ of values has relational structure in the form of a connected graph $G=(V,E)$. In the \emph{monophonic approximate agreement problem on $G$} the task is to output a set of vertices that satisfy 
\begin{itemize}
    \item agreement: the set $Y$ of output has diameter at most $d$ for a given $d \ge 1$,
    \item validity: each value $y \in Y$ lies on a chordless\footnote{A path is {\em chordless\/} (also known as \emph{minimal}) if there are no edges between non-consecutive vertices.} path between some input vertices $x, x' \in X$.
\end{itemize}
The above problem is a natural generalisation of approximate agreement onto graphs. 
It is easy to see that the discrete version of one-dimensional approximate agreement is just approximate agreement on a path (\figureref{fig:examples}a). If $G$ is a tree or a block graph\footnote{A graph is a \emph{block graph} if every $2$-connected component is a clique.}, then the task is to output vertices that lie on the minimal vertex set connecting all input vertices (\figureref{fig:examples}b--c). 

In the parlance of abstract convexity theory~\cite{Jamison1984Helly,Farber1986Convexity,Farber1987Local,Duchet1988Convex}, the validity condition requires that the output lies in the \emph{monophonic}, or \emph{minimal path},  or \emph{chordless path} convex hull of the input vertices. Another reasonable validity constraint would be to require the output values to lie on the \emph{shortest} paths between input vertices, i.e., in the \emph{geodesic} convex hull. We consider both variants and refer to the latter version of the problem as \emph{geodesic} approximate agreement on $G$.

\begin{figure}
    \begin{center}
    \includegraphics[page=1,scale=0.85]{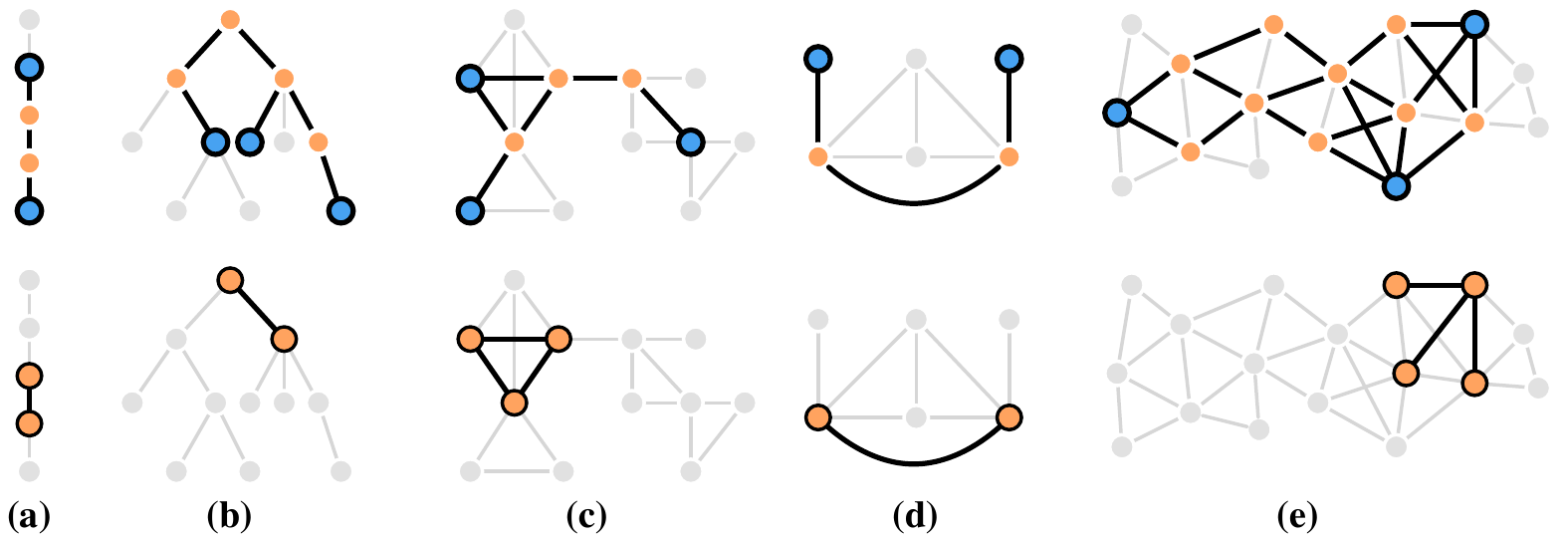}
    \end{center}
    \caption{Examples of geodesic and monophonic agreement on graphs. In the top row, the blue and orange vertices form a convex hull of the blue vertices for each graph under (a)--(d) geodesic and (e) monophonic convexities. The thick edges lie in the shortest (geodesic) or chordless (monophonic) paths between the blue vertices. The bottom row shows possible feasible outputs for the respective approximate agreement problems with $d=1$, i.e., the highlighted vertices form a clique (agreement) and are contained in the respective convex hull of the input values (validity). \label{fig:examples}}
\end{figure}

\subsection{Contributions}

In this work, we introduce the \emph{abstract approximate agreement problem on a convexity space~$\mathcal{C}$} satisfying:
\begin{itemize}
    \item agreement: $Y$ is a free set, that is, $\langle Y \rangle = \ex Y$, where $\ex Y$ is the extreme points of $Y$.
    \item validity: $Y \subseteq \langle X \rangle$.
\end{itemize}
While our primary focus lies in the graphical version of approximate agreement, we believe the abstract problem is also interesting in itself.
Indeed, it conveniently turns out that the problem coincides with various natural agreement problems: In graphs, the monophonic and geodesic approximate agreement on graphs problem given above boils down to solving approximate agreement on the chordless path or geodesic convexities of $G$. Moreover, lattice agreement on $\mathbb{L}$ is equivalent to solving approximate agreement on the \emph{algebraic convexity space} of the semilattice (sets closed under $\oplus$).  
Our key results can be summarised as follows:
\begin{enumerate}
    \item {\bf Byzantine approximate agreement on chordal graphs.} We give algorithms for approximate agreement on trees and chordal graphs. The algorithms tolerate $f < n/(\omega+1)$ Byzantine faults and terminate in $O(\log N)$ \emph{asynchronous} rounds, where $\omega$ is the clique number and $N$ is the number of vertices in the value graph $G$. In trees, we achieve optimal~resilience. 

    \item {\bf Byzantine lattice agreement on cycle-free semilattices.} As another example, we give an asynchronous lattice agreement algorithm on cycle-free lattices that tolerates up to $f < n/(\omega+1)$ Byzantine faults, where $\omega$ is the height of the semilattice. To our knowledge, this is the first algorithm that solves any variant of semilattice agreement under Byzantine~faults.
        
    \item {\bf General impossibility results for asynchronous systems.} We give impossibility results for approximate agreement on arbitrary convex geometries parameterised by two combinatorial convexity invariants: the Carath\'eodory number $c$ and the Helly number $\omega$. As corollaries, we obtain resilience lower bounds for approximate agreement problems in asynchronous systems.
        \item {\bf Optimal synchronous algorithms for convex consensus.} We consider the \emph{exact} variant of the abstract approximate agreement problem, where the agreement constraint is replaced by $|Y|=1$. While the problem cannot be solved in asynchronous systems, we show that it can be solved on any convex geometry $\mathcal{C}$ in $\Theta(f)$ synchronous rounds if and only if $n > \omega f$ holds, where $\omega$ is the Helly number of $\mathcal{C}$. Moreover, the upper bound holds for \emph{any} convexity~space.
\end{enumerate}
Our work can be seen as an extension of the Mendes--Herlihy approximate agreement and Vaidya--Garg multidimensional consensus frameworks~\cite{Mendes2013Multidimensional,vaidya13byzantinevector,MHVG15} onto general convexity spaces. However, while these operate in continuous $m$-dimensional Euclidean spaces, our analysis relies on combinatorial theory of abstract convexity, where the input and output values have discrete, combinatorial structure. In particular, the discrete nature of the convexity space poses new challenges, as unlike in the continuous setting, non-trivial convex sets do not necessarily contain non-extreme points to choose from to facilitate convergence. 

Multidimensional agreement problems in Euclidean spaces have applications ranging from, e.g., robot convergence tasks to distributed voting and convex optimisation~\cite{MHVG15}. Our work extends the scope of these techniques to discrete convexity spaces, which can be used to describe various natural combinatorial systems. 
Finally, unlike prior work, our algorithms do not assume that processors can perform computations or send messages involving arbitrary precision real values, as in the discrete case a single value can be encoded using $O(\log |V|)$ bits.

\subsection{Related work}

The seminal result of Fischer et al.~\cite{fischer85impossibility} showed that \emph{exact} consensus cannot be reached in asynchronous systems in the presence of crash faults. 
Dolev et al.~\cite{Dolev1986Reaching} showed that it is however possible to reach \emph{approximate agreement} in an asynchronous system even with arbitrary faulty behavior when the values reside on the continuous real line. 
Subsequently, the one-dimensional approximate agreement problem has been extensively studied~\cite{Dolev1986Reaching,fekete1990asymptotically,fekete1994asynchronous,Abraham2005optimal}. Fekete~\cite{fekete1994asynchronous} showed that any algorithm reducing the distance of values from $d$ to $\varepsilon$ requires $\Omega(\log (\varepsilon / d))$ asynchronous rounds when $f \in \Theta(n)$; in the discrete setting this yields the bound $\Omega(\log N)$ for paths of length $N$. Recently, Mendes et al.~\cite{MHVG15} introduced the natural generalisation of \emph{multidimensional} approximate agreement and  showed that the $m$-dimensional problem is solvable in an asynchronous system with Byzantine faults if and only if $n > (m+2)f$ holds for any given $m \ge 1$. 

The \emph{lattice agreement problem} was originally introduced in the context of wait-free algorithms in shared memory models~\cite{attiya1995atomic,attiya2001adaptive}. The problem has recently resurfaced in the context of asynchronous message-passing models with crash faults~\cite{faleiro2012generalized,zheng2018lattice}. These papers consider the problem when the validity condition is given as $x_i \le y_i \le \bigoplus \{ x_j : j \in P \}$, i.e., the output of a processor must satisfy $x_i \le y_i$ and the feasible area is determined also by the inputs of faulty processors. However, it is not difficult to see that under Byzantine faults, this validity condition is not reasonable, as the problem cannot be solved even with one faulty processor.

Another class of structured agreement problems in the wait-free asynchronous setting are \emph{loop agreement} tasks~\cite{herlihy2003classification}, which generalise $k$-set agreement and approximate agreement (e.g., $(3,2)$-set agreement and one-dimensional approximate agreement). In loop agreement, the set of inputs consists of three distinct vertices on a loop in a 2-dimensional simplicial complex and the outputs are vertices of the complex with certain constraints, whereas \emph{rendezvous tasks} are a generalisation of loop agreement to higher dimensions~\cite{liu2009rendezvous}. These tasks are part of large body of work exploring the deep connection of asynchronous computability and combinatorial topology, which has successfully been used to characterise the \emph{solvability} of various distributed tasks~\cite{herlihy2013book}.
Gafni and Kuznetsov's $P$-reconciliation task~\cite{gafni2007object} achieves geodesic approximate agreement on a graph of system configurations.

Finally, we note that distributed agreement tasks play a key role in many fault-tolerant clock synchronisation algorithms~\cite{welch1988new,lenzen19pulse,lenzen2018firing}. Byzantine-tolerant clock synchronisation can be solved using one-dimensional approximate agreement~\cite{welch1988new}, whereas in the \emph{self-stabilising} setting both exact digital clock synchronisation~\cite{lenzen2018firing} and pulse synchronisation tasks reduce to consensus~\cite{lenzen19pulse}. However, while the latter problem has been extensively studied~\cite{dolev04clock-synchronization,dolev14fatal,lenzen19pulse}, non-trivial lower bounds are still lacking~\cite{lenzen19pulse}. Given that clock synchronisation closely relates to agreement on cyclic structures, investigating agreement tasks on different structures may yield insight into the complexity of fault-tolerant (approximate) clock synchronisation. Indeed, we show that agreement on graphs \emph{without} long induced cycles is considerably easier than consensus. 

\section{Preliminaries}\label{sec:preliminaries}

We start with some basic preliminaries needed to describe the main ideas and results of the paper.

\subsection{Abstract convexity spaces}

Let $V$ be a finite set. The collection $\mathcal{C} \subseteq 2^V$ is a \emph{convexity space on $V$} if 
     (1) $\emptyset, V \in \mathcal{C}$ holds, and (2) 
    $A,B \in \mathcal{C}$ implies that $A \cap B \in \mathcal{C}$.
A set $K \in \mathcal{C}$ is said to be \emph{convex}. For any $A \subseteq V$, the \emph{convex hull} of $A$ is the minimal convex set $\langle A \rangle \in \mathcal{C}$ such that~$A \subseteq \langle A \rangle$. Thus, $\langle \cdot \rangle$ is a closure operator on $V$. 
For any $A \subseteq V$ and $a \in A$, $a$ is called an \emph{extreme point} of $A$ if $a \notin \langle A \setminus a \rangle$. For a convex set $K \in \mathcal{C}$, we use $\ex K$ to denote the extreme points of $K$. The convexity space $\mathcal{C}$ is a \emph{convex geometry} if every $K \in \mathcal{C}$ satisfies $K = \langle \ex K \rangle$.
A convex set $K$ is \emph{free} if $K = \ex{K}$. 
Finally, a nonempty set $A \subseteq V$ is \emph{irredundant} if
$
 \partial A \neq \emptyset$ where $\partial A = \langle A \rangle \setminus \bigcup_{a \in A} \langle A \setminus a \rangle$.
The following theorem characterises convex geometries:
\begin{theorem}[\cite{Edelman1985}]\label{thm:cg-equivalences}
    Let $\mathcal{C}$ be a convexity space on $V$. The following conditions are equivalent:
    \begin{enumerate}
    	\item $\mathcal{C}$ is a convex geometry.
        \item For every $K \in \mathcal{C}$, $K = \langle \ex K \rangle$ (Minkowski-Krein-Milman property).
        \item For every $K \in \mathcal{C} \setminus \{ V \}$, there exists an element $u \in V \setminus K $ such that $K \cup \{u \} \in \mathcal{C}$.
    \end{enumerate}
\end{theorem}

\subsubsection{Carath\'eodory and Helly numbers} 
The \emph{Carath\'eodory number} of a convexity space $\mathcal{C}$ on $V$ is the smallest integer $c$ such that for any $U \subseteq V$ and any $u \in \langle U \rangle$, there is a set $S \subseteq U$ such that $|S| \le c$ and $u \in \langle S \rangle$. The Carath\'eodory number of a convexity space equals the maximum size of an irredundant set in $\mathcal{C}$. A collection $\mathcal{C}$ of sets is $k$-intersecting if every $\mathcal{B} \subseteq \mathcal{C}$ with $|\mathcal{B}| \le k$ has a nonempty intersection. The \emph{Helly number} of a convexity space $\mathcal{C}$ is the smallest integer~$\omega$ such that any finite $\omega$-intersecting $\mathcal{A} \subseteq \mathcal{C}$ has a nonempty intersection. If $\mathcal{C}$ is a convex geometry, the Helly number equals the maximum cardinality of a free set in $\mathcal{C}$~\cite{Edelman1985}.

\subsubsection{Examples of convex geometries} In this work, we focus on the following convexity spaces:
\begin{itemize}
    \item Let $G = (V,E)$ be a graph. A set $U \subseteq V$ is convex if all the vertices on all minimal, i.e., chordless, paths connecting any $u,v \in U$ are contained in $U$. The Helly number of this convexity space equals the size of the maximum clique in $G$~\cite{Jamison1984Helly,Duchet1988Convex} and the Carath\'eodory number is at most two~\cite{Duchet1988Convex}. Moreover, free sets coincide with cliques. The convexity space is a convex geometry iff $G$ is chordal~\cite{Farber1986Convexity}. Indeed, if $G$ is a cycle of length at least four, then it is easy to check that $V$ is convex, but has no extreme points. 

    \item Let $\mathbb{L} = (V, \oplus)$ be a semilattice and $\mathcal{C} \subseteq 2^V$ be a collection of subsets closed under $\oplus$. The collection $\mathcal{C}$ is a convex geometry, where every $K \in \mathcal{C}$ is a subsemilattice $(K, \oplus)$ of $\mathbb{L}$. A set $K$ is free if and only if it is a chain~\cite{poncet2014semilattice}. Thus, the Helly number of a semilattice equals its height. Moreover, the Carath\'eodory number equals the breadth of the semilattice.
\end{itemize}

\subsection{Asynchronous rounds}

When operating in the asynchronous model, we describe and analyse the algorithms in the \emph{asynchronous round} model. In this model, each processor has a local round counter and labels all of its messages with a round number. Each correct node initialises its round counter to 0 at the start of the execution and increases its local round counter from $t$ to $t+1$ only when it has received at least $n-f$ messages belonging to round $t$ (since up to $f$ faulty nodes may omit their messages). In each asynchronous round $t \ge 0$, a non-faulty processor $i \in P \setminus F$
\begin{enumerate}[noitemsep]
  \item sends a value to each processor $j \in P$,
  \item receives a value $M_{ij}(t)$ from each processor $j \in P$,
  \item updates local state and proceeds to round $t+1$.
\end{enumerate}
The received message $M_{ij}(t)$ may be empty, denoted by $\emptysym$, to indicate that no message arrived from processor $j$ (e.g., due to a crash or a delay). We use the set 
\[
 P_i(t) = \{ j \in P : M_{ij}(t) \neq \emptysym \}
\]
to denote the processors from which $i$ received a nonempty message on round $t$. 

Assuming $n>3f$ and with the help of reliable broadcast, the witness technique~\cite{Abraham2005optimal,MHVG15}, and attaching round numbers to all messages,
the Byzantine asynchronous round model can guarantee the following for each $i,j \in P \setminus F$:
\begin{enumerate}[noitemsep]
    \item $|P_i(t)| \ge n-f$,
    \item $|P_i(t) \cap P_j(t)| \ge n-f$,
    \item if $M_{ik}(t) = x \neq \emptysym$ for $k \in F$, then $M_{jk}(t) \in \{ x, \emptysym \}$.
\end{enumerate}
That is, (1) every correct processor receives at least $n-f$ nonempty values (out of which $f$ may be from faulty processors), (2) any two correct processors receive at least $n-f$ common values (possibly $f$ of which may be from faulty processors), and (3) if some correct processor receives a nonempty value $x$ from a faulty processor, then all other correct processors receive the same value or no value. 
The Byzantine asynchronous round model can be simulated in the asynchronous model so that a non-faulty processor broadcasts $O(n \log n)$ additional bits per round~\cite{Abraham2005optimal,MHVG15}.

\subsection{Graphs}
Let $G = (V,E)$ be a finite undirected connected graph, where $V=V(G)$ denotes the set of vertices and $E=E(G)$ the set of edges. We assume all graphs are simple (no parallel edges) and loopless (no self-loops). For any $U \subseteq V$, we use $G[U] = (U, F)$, where $F = \left\{ \left\{u,v \right \} \in E : u,v \in U \right \}$, to denote the subgraph of $G$ induced by the vertices in $U$. An $\ell$-length path $u \leadsto v$ from a vertex $u$ to vertex $v$ is a non-repeating sequence $(u=v_0, \ldots, v_\ell = v)$ of vertices such that $\{v_i,v_{i+1}\} \in E$. An $\ell$-cycle is an $(\ell-1)$-path from $u$ to $v$ with $\{u,v\} \in E$. A path $v_0, \ldots, v_\ell$ is \emph{chordless} (or minimal) if there does not exist any edge $\{ v_{i}, v_{j} \} \in E$ for $j > i+1$. 

For vertices $u,v \in V$, we denote the length of the shortest path between $u$ and $v$ as $d(u,v)$. The eccentricity of vertex $v$ is $\ecc(v) = \max \{ d(v,u) : u \in V(G) \}$. For a set $U \subseteq V$, we define its diameter $D(U) = \max \{ d(u,v) : u,v \in U \}$. The diameter of graph $G$ is denoted by $D(G) = D(V)$. The \emph{radius} of a graph $G$ is $R(G) = \min \{ \ecc_G(v) : v \in V(G) \}$ and the \emph{center} is $\cent(G) = \{ u \in V(G) : \ecc_G(u) = R(G) \} \subseteq V(G)$. For a connected set $U \subseteq V$, we use the short-hands $R(U) = R(G[U])$ and $\cent(U) = \cent(G[U])$.

A graph $G$ is a \emph{tree} if it contains no cycles and \emph{chordal} if contains no $\ell$-cycle with $\ell \ge 4$ as an induced subgraph. A graph is \emph{Ptolemaic} if it is chordal and distance-hereditary (any connected induced subgraph preserves distances). The clique number $\omega(G)$ is the size of the largest clique in $G$. 
A vertex is \emph{simplicial} in $U \subseteq V$ if its neighbourhood $\mathcal{N}(v) \cap U$ in $U$ is a clique. A perfect elimination ordering $\preceq$ of $G = (V,E)$ is a total order on $V$ such that any $u \in V$ is simplicial in $\{ v : u \preceq v \}$. A graph has a perfect elimination ordering iff it is chordal.

\subsubsection{Chordal graphs} Chordal graphs (also known as triangulated, rigid or decomposable graphs) is an important and well-studied class of graphs. From a structural point of view, they have many equivalent characterisations: they are graphs that have no induced cycles greater than three, graphs for which perfect elimination orderings exist, graphs in which every minimal vertex separator is a clique, and others~\cite{dirac1961rigid,rose1970triangulated,gavril1972algorithms,rose1976algorithmic,Farber1986Convexity}.

Due to their ubiquitous nature and convenient structural properties, the algorithmic aspects of chordal graphs have received much attention in the past decades. For example, chordal graphs have applications in a variety of contexts including combinatorial and semidefinite optimisation~\cite{vandenberghe2015chordal} and probabilistic graphical models~\cite{lauritzen1996graphical}. Indeed, many NP-hard problems, such as finding maximum cliques or optimal vertex colourings, often admit simple polynomial time solutions in chordal graphs~\cite{gavril1972algorithms}. In the distributed setting, it is possible to find good approximations to minimum vertex colourings and maximum independent sets in chordal graphs~\cite{Konrad2018Chordal}.

\subsection{Lattices}
A (join) semilattice is an idempotent semigroup $\mathbb{L} = (V, \oplus)$,
where $V$ is a finite set and $\oplus : V \times V \to V$ is called the
\emph{join} operator. A semilattice has a natural partial order defined as $u
\le v \iff u \oplus v = v$, where $u \oplus v$ is the least upper bound (i.e.,
a join) of $\{u,v\}$ in the partial order. We write $u < v$ if $u \le v$ and $u
\neq v$. For any set $U = \{u_0, \ldots, u_\ell \} \subseteq V$, the least
upper bound $\bigoplus U$ is known as the join of $U$. If $u_0 \le \ldots \le
u_\ell$ holds, then $U$ is said to be a chain from $u_0$ to $u_\ell$. 
If $\{u,v\}$ is a chain, then $u$ and $v$ are
said to be comparable. 
The height $\omega(\mathbb{L})$ of a semilattice is the maximum cardinality of any chain $U \subseteq V$.
The breadth of a semilattice is the smallest integer $b$ such that for any nonempty $U \subseteq V$ there is a subset $A \subseteq U$ of size at most $b$ satisfying $\bigoplus A = \bigoplus U$.

\section{Approximate agreement on abstract convexity spaces}\label{sec:helly}

\subsection{Iterative algorithm on abstract convexity spaces\label{sec:algorithm}}

In this section, we describe a basic step for an approximate agreement algorithm in the Byzantine asynchronous round model in an abstract convexity space $\mathcal{C}$ with Helly number $\omega$. The algorithm is a generalisation of the Mendes--Herlihy algorithm by Mendes et al.~\cite{MHVG15} onto abstract convexity spaces. It is not guaranteed, however, to converge on \emph{all} discrete convexity spaces.

The algorithm proceeds iteratively. At the start of each asynchronous round $t$, each correct processor $i \in P \setminus F$ broadcasts its current value $x_i(t) \in V$. At the end of round $t$, processor $i$ has received a value from at least $n-f$ processors $P_i(t) \subseteq P$. These values are used to compute a new value $x_i(t+1) = y_i(t)$. For brevity, we often omit $t$ from our notation, e.g., use the short-hands such as $P_i = P_i(t)$. 

\subsubsection{Computing the safe area} 
For any subset of processors $J \subseteq P_i$, define
\[
 \mathcal{V}_i(J) = \{ M_{ij}(t) \neq \emptysym : j \in J \} 
\]
to be the set of values processor~$i$ received from processors in $J$. Processor $i$ locally computes 
\[
 \mathcal{K}_i = \left\{ \langle \mathcal{V}_i(J) \rangle : J \in {P_i \choose |P_i|-f} \right\} \quad \textrm{ and } \quad H_i = \bigcap \mathcal{K}_i,
\]
where $\langle \cdot \rangle$ denotes the convex hull operator. Processor $i$ then outputs the value 
\[
 y_i = \begin{cases}
     \phi(H_i) & \textrm{if } H_i \neq \emptyset, \\
     \emptysym & \textrm{otherwise,}
 \end{cases}
 \]
where $\phi \colon \mathcal{C} \to V$ is an output map, which will depend on the convexity space $\mathcal{C}$ we are working in, see Section~\ref{sec:chordal} for an output map for chordal graphs. The Helly property guarantees that $H_i$ and $y_i$ remain in the closure  $\langle X \rangle$ of the input values. 
For each $t \ge 0$, we define $X(t) = \{ x_j(t) : j \in P \setminus F \} \subseteq V$ and $Y(t) = \{ y_j(t) = \phi(H_j(t)) : j \in P \setminus F \}$.

\begin{lemma}\label{lemma:validity-combined}
    Suppose $\mathcal{C}$ is a convexity space on $V$ with Helly number $\omega$. If $n > \max\{ (\omega+1)f, 3f\}$ holds, then for each iteration $t \ge 0$ the above algorithm satisfies
    \begin{itemize}
        \item $\emptyset \neq H_i(t) \subseteq \langle X(t) \rangle$ for all $i \in P \setminus F$,
        \item $\bigcap_{i \in P \setminus F } H_i(t) \neq \emptyset$.
    \end{itemize}
\end{lemma}

\subsubsection{Proof of Lemma~\ref{lemma:validity-combined}}\label{apx:validity-combined}

We divide the proof of \lemmaref{lemma:validity-combined} into smaller lemmas.

\begin{lemma}\label{lemma:validity}
 Let $X = \{ x_j : j \in P \setminus F \}$. If $n > 3f$, then $H_i \subseteq \langle X \rangle$.
\end{lemma}
\begin{proof}
  Observe that some $J \in {P_i \choose |P_i| - f}$ satisfies $J \subseteq P_i \setminus F$. As $|P_i|-f \ge n - 2f > f \geq 0$, we have that $J \neq \emptyset$. Since $\langle \mathcal{V}_i(J) \rangle \in \mathcal{K}_i$, it follows that $H_i = \bigcap \mathcal{K}_i \subseteq \langle \mathcal{V}_i(J) \rangle \subseteq \langle X \rangle$.
\end{proof}

\begin{lemma}\label{lemma:intersection-size}
 If $\mathcal{A} = \{A_1, \ldots, A_k\}$, $\mathcal{B} = \{ B_1, \ldots, B_k \}$, and $\mathcal{D} = \{ D_1, \ldots, D_k \}$ are collections of sets such that $A_i = B_i \setminus D_i$, then $|\bigcap \mathcal{A} | \ge |\bigcap \mathcal{B} | - \sum_{1 \le i \le k} \lvert D_i\rvert$.
\end{lemma}
\begin{proof}
Since $|A \setminus B| \ge |A| - |B|$ for any sets $A,B$, the claim follows by observing that 
    \[
    \bigcap_{1 \le i \le k} A_i = \bigcap_{1 \le i \le k} (B_i \setminus D_i) 
                                = \bigcap_{1 \le i \le k} B_i \setminus \bigcup_{1 \le i \le k} D_i.\qedhere
\]
\end{proof}

\begin{lemma}\label{lemma:nonempty}
    If $n > \max\{ (\omega+1)f , 3f \}$, then $\bigcap_{i \in P \setminus F} H_i \neq \emptyset$.
\end{lemma}
\begin{proof}
    We prove the claim by showing that the collection 
\[
\mathcal{K} = \bigcup_{i \in P \setminus F} \mathcal{K}_i
\]
has a nonempty intersection by establishing that it is $\omega$-intersecting. Let $\mathcal{A} = \{ A_1, \ldots, A_\omega \} \subseteq \mathcal{K}$. By definition $A_k \in \mathcal{K}_{\tau(k)}$ for some $\tau(k) \in P \setminus F$ and 
\[
    A_k = \langle \mathcal{V}_{\tau(k)}(J_k) \rangle, \textrm{ where } J_k \in {P_{\tau(k)} \choose |P_{\tau(k)}|-f}.
\]
Recall that the asynchronous round model guarantees that all correct processors receive at least $n-f$ common values. Applying \lemmaref{lemma:intersection-size} to $\mathcal{P} = \{P_{\tau(1)}, \ldots, P_{\tau(\omega)}\}$ and $\mathcal{J} = \{ J_1, \ldots, J_\omega \}$ yields
    \begin{align*}
    \lvert \bigcap \mathcal{J}\rvert \ge \lvert\bigcap \mathcal{P}\rvert - \lvert\sum_{k=1}^\omega P_{\tau(k)} \setminus J_k\rvert
                          \ge n-f - \omega f 
                          > (\omega+1)f - f - \omega f = 0.
    \end{align*}
    Hence, there is some $j \in \bigcap \mathcal{J}$, for which $\mathcal{V}_{\tau(k)}(\{j\}) \in \bigcap \mathcal{A}$.
Since $\mathcal{K} \subseteq \mathcal{C}$ is $\omega$-intersecting and the convexity space $\mathcal{C}$ has a Helly number of $\omega$, the claim follows. 
\end{proof}

\subsection{On the elimination of extreme points}

If we can in each iteration remove some extreme point of $\langle X \rangle$, where $X$ is the set of input values, then the hull of output values $\langle Y \rangle$ shrinks. In an arbitrary convexity space, a convex set may not have any extreme points (consider, e.g., the chordless path convexity on a four cycle). However, in a convex geometry every nonempty convex set has an extreme point, as $\langle X \rangle = \langle \ex X \rangle$ by \theoremref{thm:cg-equivalences}.

Moreover, finite convex geometries are in a sense ``easy to peel'' iteratively -- at least from the global perspective. We say that a total order $\preceq$ on $V$ is a \emph{convex elimination order} of $\mathcal{C}$ if for any $u \in V$ the sets $A(u) = \{ v \in V : u \preceq v \}$ and $A(u) \setminus  u$ are convex. \theoremref{thm:cg-equivalences} implies that convex geometries admit such orderings and  we assume that the values in $V$ are labelled according to such an order~$\preceq$.
For a set $U \subseteq V$, we let $\max U$ and $\min U$ be the unique maximal and minimal elements, respectively, given by $\preceq$ such that for all $v \in U$ we have $\min U \preceq v $ and $v \preceq \max U$. 
\begin{remark}\label{remark:min-is-ex}
For any $K \subseteq V$ it holds that $\min K \in \ex K$ and $\min K = \min \langle K \rangle$.
\end{remark}

The next lemma shows that to guarantee progress by shrinking the set of output values, it suffices to always exclude, e.g., $\min X$ from the output (of course, this is not indefinitely possible). 
\begin{lemma}\label{lemma:without-min-progress}
    If $\min X \notin Y$ in a convex elimination order, then $\langle Y \rangle \subsetneq \langle X \rangle$.
\end{lemma}
\begin{proof}
    Recall that $Y \subseteq \langle X \rangle$ by \lemmaref{lemma:validity-combined} and $\min X \in \ex X$ by \remarkref{remark:min-is-ex}. Since the closure operator is increasing and $\min X \notin Y$, we get that $\langle Y \rangle \subseteq \langle X \rangle \setminus \{ \min X \} \subseteq \langle X \setminus \{ \min X \} \rangle \subsetneq \langle X \rangle$.
\end{proof}

\section{Approximate agreement on chordal graphs}\label{sec:chordal}

We now show that monophonic approximate agreement on chordal graphs can be solved given that $n > (\omega+1)f$ holds, where $\omega$ is the clique number of $G$. This also implies that geodesic approximate agreement is solvable on Ptolemaic graphs. 
Throughout we assume that $G = (V,E)$ is a connected chordal graph with at least two vertices and $\mathcal{C}$ is its chordless path convexity space. We recall that the Helly number of $\mathcal{C}$ coincides with the clique number $\omega=\omega(G)$ of $G$~\cite{Jamison1984Helly,Duchet1988Convex}. 

\subsection{Approximate agreement on trees}

Suppose $G=(V,E)$ is a tree. As $G$ is also chordal, it has a perfect elimination ordering~$\preceq$. We assume that the vertices of $V$ are labelled according to this ordering and define the output map
\[
\phi(K) = \max \cent K,
\]
where $\cent K \subseteq V$ is the center of the subgraph $G[K]$ induced by $K$.
This rule roughly divides the diameter of the convex hull of active values by two.

We start by showing that the above rule handles trees of diameter two, i.e., star graphs. Here, the perfect elimination ordering is used for symmetry-breaking to guarantees convergence onto an edge.

\begin{lemma}\label{lemma:star-step}
 If $\langle X \rangle$ has diameter two, then $Y$ has diameter one.
\end{lemma}
\begin{proof}
    As $\langle X \rangle$ has diameter two, it is a star. Since $H_i \subseteq \langle X \rangle$ for all $i \in P \setminus F$, there exists a vertex
    \[
    v \in \cent \langle X \rangle \cap \bigcap_{i \in P \setminus F} \cent H_i,
    \]
    where $v$ is adjacent to all $u \in \langle X \rangle$ where $u \neq v$. Now $v \preceq y_i = \phi(H_i) = \max \cent H_i$. As $\preceq$ is a perfect elimination order, we get that $Y \subseteq \{ u : \{u,v\} \in E, v \preceq u \}$ is a clique.
\end{proof}

\begin{remark}\label{remark:tree-diameter-radius}
  In any graph $G$, the diameter and radius satisfy $D(G) \le 2R(G)+1$.
\end{remark}

\begin{lemma}\label{lemma:tree-reduction}
    If $n > 3f$ and $G=(V,E)$ is a tree, then a single iteration satisfies $D(Y) \le D(X)/2 + 1$.
\end{lemma}
\begin{proof}
  To prove the lemma, we show the following claim: for any $y_i, y_j \in Y$ we have $d(y_i,y_j) \le R(H_i \cup H_j) + 1$.
  To see why this is sufficient, observe that
  \[
  d(y_i, y_j) \le R(H_i \cup H_j) + 1 \le D(H_i \cup H_j)/2 + 1 \le D(\langle X \rangle) / 2 + 1.
  \]
  In particular, this bound holds for the pair $y_i, y_j$ with $d(y_i,y_j) = D(\langle Y \rangle)$.

  Let $i,j \in P \setminus F$. We distinguish two cases. First, observe that if $H_j \subseteq H_i$ holds, then $y_i \in \cent H_i$, which implies that $d(y_i,y_j) \le R(H_i) + 1 = R(H_i \cup H_j) + 1$ proving the claim. For the second case, suppose that neither $H_j \subseteq H_i$ or $H_i \subseteq H_j$ hold. Without loss of generality, we can assume $d(y_i, y_j) > 1$. Let $B = \{ v_0, \ldots, v_m \}$ be the vertices on the unique path of length $m-1$ between $y_i$ and $y_j$ (excluding the end-points $y_i$ and $y_j$). Consider the two disjoint trees $T_i$ and $T_j$ that comprise $G[ (H_i \cup H_j) \setminus B]$. Since for both $k \in \{i,j\}$ we have $y_k \in \cent H_k$, we can choose $u_k \in T_k$ such that $d(u_k, y_k) = r(H_k)$ and $u_k \notin B$. Moreover, both $y_i$ and $y_j$ lie on the path connecting $u_i$ and $u_j$ in~$G$.
Using the fact that $R(H_i \cup H_j) \le R(H_i) + R(H_j)$  we get
  \begin{align*}
    D(H_i \cup H_j) & \ge d(u_i, u_j) \\
    &= d(u_i, y_i) + d(y_i, y_j) + d(u_j, y_j) \\
    &= R(H_i) + R(H_j) + d(y_i, y_j) \\
    &\ge R(H_i \cup H_j) + d(y_i, y_j),
  \end{align*}
  which together with \remarkref{remark:tree-diameter-radius} implies that
  \begin{align*}
    d(y_i, y_j) &\le D(H_i \cup H_j) - R(H_i \cup H_j) \\
    & \le R(H_i \cup H_j) + 1. \qedhere
  \end{align*} 
\end{proof}

\begin{theorem}\label{thm:trees}
    If $n > 3f$ and $G=(V,E)$ is a tree, then approximate agreement on $G$ can be solved in $\log D(G)+1$ asynchronous rounds, where $D(G)$ is the diameter of $G$.
\end{theorem}
\begin{proof}
    The validity condition is trivially satisfied as $Y(t) \subseteq \langle X(t) \rangle$ for any $t \ge 0$ by \lemmaref{lemma:validity-combined}. Now we show by induction on $t$ that $D(X(t)) \le D(X(0))/2^t + 1$ for $t \le \lceil \log D(G) \rceil$. The base case $t=0$ is trivial. Suppose the claim holds for some $t < \log D(G)$. From \lemmaref{lemma:tree-reduction}, we get 
    \[
    D(X(t+1)) = D(Y(t)) \le \frac{D(X(t))}{2} + 1 \le \frac{D(X(0))/2^t + 1}{2} + 1 \le \frac{D(X(0))}{2^{t+1}} + 3/2.
    \]
    Since the diameter is integral, the diameter satisfies $\frac{D(X(0))}{2^{t+1}} +1 \le \frac{D(G)}{2^{t+1}} + 1$. Running the algorithm for $t = \lceil \log D(G) \rceil + 1$ iterations yields that the diameter of $X(t)$ is at most two on iteration $t$. 
    If $X(t)$ has diameter 1, the claim follows. Otherwise, the claim follows from \lemmaref{lemma:star-step}.
\end{proof}

\subsection{Fast monophonic approximate agreement on chordal graphs}

We now  present a fast monophonic approximate agreement algorithm on chordal graphs. 
To this end, we use the tree algorithm above on a suitable tree decomposition of the actual graph~$G$.

\begin{definition}
    Let $G$ be a graph, $T$ a tree and $\chi \colon V(T) \to 2^{V(G)}$ be a mapping. We say that the pair $(T, \chi)$ is a \emph{tree decomposition} of $G$ if the following conditions are satisfied:
    \begin{itemize}
        \item for all $v \in V(G)$ there exists $b \in V(T)$ such that $v \in \chi(b)$, 
        \item for all $e \in E(G)$ there exists $b \in V(T)$ such that $e \subseteq \chi(b)$, 
        \item if $v \in \chi(a) \cap \chi(b)$, then $v \in \chi(c)$ for all $c \in V(T)$ residing on the unique path $a \leadsto b$.
    \end{itemize}
The tree decomposition is a \emph{clique tree} if each $b \in V(T)$ induces a maximal clique $\chi(b)$ in~$G$.
\end{definition}

Chordal graphs can be characterised as those graphs having a clique tree~\cite[Proposition 12.3.11]{Diestel2010Textbook}.
In fact, if~$G$ is chordal, the tree~$T$ can always be chosen as a spanning tree of $G$'s clique graph, i.e., the graph whose nodes are the maximal cliques of~$G$ and whose edges join those cliques with nonempty intersection~\cite[Theorem 3.2]{Blair1993Cliquetrees}.
Unlike non-chordal graphs in which the number of maximal cliques can be exponential, the number of maximal cliques is at most linear in chordal graphs:

\begin{lemma}[{\cite{Berry2011Cliques}}]\label{lemma:cliquetree-upper}
If $G$ is a chordal graph, then it has a clique tree $(T,\chi)$ with $\lvert V(T)\rvert \in O(\lvert V(G)\rvert)$.
\end{lemma}

For the purposes of our algorithm, we use a special kind of clique trees, which we call \emph{expanded}. A clique tree $(T, \chi)$ is \emph{expanded} if for each $\{a,b\} \in E(T)$ we have either $\chi(a) \subseteq \chi(b)$ or $\chi(b) \subseteq \chi(a)$; see \figureref{fig:tree-decomp}a--b for an example of a expanded clique tree.

\begin{figure}
    \begin{center}
    \includegraphics[page=3,scale=0.8]{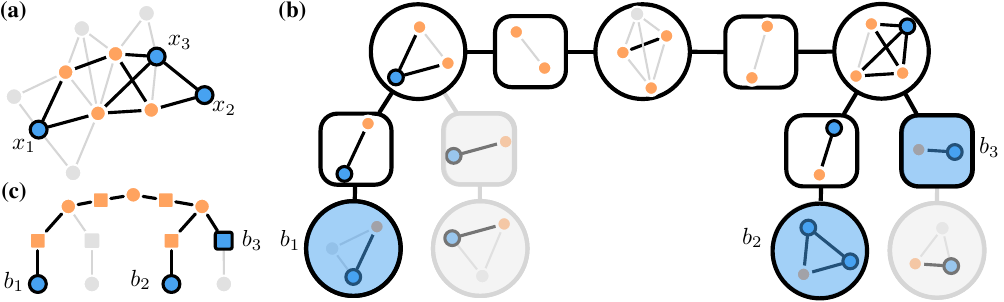}
    \end{center}
    \caption{Approximate agreement on chordal graphs via clique trees. (a) The chordal value graph $G$ and three input values $X = \{ x_1,x_2,x_3 \} \subseteq V(G)$. (b) An expanded clique tree $(T, \chi)$ of $G$. The bags $B = \{b_1,b_2,b_3\}$ satisfy $x_i \in \chi(b_i)$ and the bags are used as input for the approximate agreement algorithm on trees. The round bags are maximal cliques and the rectangular bags are the intersection of its neighbouring round bags, which are the minimal vertex separators in $G$. Note that a bag $b \in \langle B \rangle$ may contain vertices of $G$ outside the convex hull $\langle X \rangle$ of the initial input values $X$ (e.g., bag $b_1$ and the central bag). (c)~The graph $T$ on which we run the tree~algorithm.\label{fig:tree-decomp}}
\end{figure}

\begin{lemma}\label{lemma:expanded:upper}
Every chordal graph~$G$ has an expanded clique tree~$T$ with $\lvert V(T)\rvert = O(\lvert V(G)\rvert)$.
\end{lemma}
\begin{proof}An expanded clique tree can be constructed from from any clique tree $(T,\chi)$ of $G$ as follows:
\begin{align*}
    V(T') = V(T) \cup E(T) \quad \textrm{ and } \quad E(T') = \bigcup_{\substack{e = \{b_0,b_1\} \in E(T)}} \big\{ \{b_0,e\}, \{e,b_1\} \big\}
    \,.
\end{align*}
The expanded map $\chi' : V(T') \to V(G)$ is defined as
\[
 \chi(b) = \begin{cases}
    \chi(b) & \textrm{if } x \in V(T),\\
     \bigcap_{x \in b } \chi(x) & \textrm{if } x \in E(T).
 \end{cases}
 \]
The upper bound on $\lvert V(T)\rvert$ follows from \lemmaref{lemma:cliquetree-upper}.
\end{proof}

\subsubsection{The algorithm}
Let $G = (V,E)$ denote the chordal value graph, $(T, \chi)$ an expanded clique tree of~$G$, and $\mathbf{A}$ the approximate agreement algorithm on trees given by \theoremref{thm:trees}. 
Given an input $x_i(0) \in V(G)$ on the graph $G$, processor $i \in P \setminus F$ starts by choosing any \emph{bag} $b_i(0) \in V(T)$ such that the $x_i(0) \in b_i(0)$; see \figureref{fig:tree-decomp}. 
In iteration $t \geq 1$, every processor $i \in P \setminus F$ performs the following:
\begin{enumerate}
    \item Broadcast $x_i(t)$ and $b_i(t)$ to all other processors.
    \item Simulate one step of $\mathbf{A}$ on the $b(\cdot)$ values and set $b_i(t+1) = \mathbf{A}\left(b_{0,i}(t), \ldots, b_{n-1,i}(t)\right)$.
    \item Compute the safe area $H_{i}^G$ from the received values $x_{ij}(t)$.
    \item Set $x_i(t+1)$ to an arbitrarily chosen element of $\chi(b_i(t+1)) \cap H_i^G$.
\end{enumerate}

Since the $b_i(\cdot)$ values are updated using the algorithm $\mathbf{A}$, these values converge onto a single edge $\{a,b\} \in E(T)$ in the tree $T$.
As $\chi(a) \cup \chi(b)$ is a clique due to the expandedness of~$T$, the output values $x(\cdot)$ will have diameter at most one in $G$ assuming that $x_i(t+1)$ is well-defined for each $i \in P_i \setminus F$, that is, $\chi(b_i(t+1)) \cap H_i^G \neq \emptyset$. Showing this is the main challenge of the correctness~proof.

Consider processor $i \in P \setminus F$ and let $t \ge 0$.
We will show that $b\cap H_i^G \neq \emptyset$ for all $b\in H_i^T$, where $H_i^T$ is the intersection of convex sets computed by algorithm $\vec A$.
For each $v \in V(G)$ we define $\rho(v)$ to be the set of pairs $(j,k)\in P_i\times P_i$ such that some chordless path from~$x_j$ to~$x_k$ visits vertex~$v$. We say that a vertex in the tree $T$ is an interior vertex if it is not a leaf in $T$.

\begin{lemma}\label{lemma:intersecting-paths}
Let $b \in H_i^T$ be an interior node of~$H_i^T$.
There exists 
 $v \in \chi(b)$ such that $\lvert\rho(v)\vert \geq f+1$.
\end{lemma}
\begin{proof}
Let $T_1, \dots, T_m$ be the branches of~$T$ when rooted at~$b$.
Since~$b$ is not a leaf, we have $m \geq 2$.
Define the sets of vertices $V_0 = \chi(b)$ of graph $G$ and
\[
V_r
=
\bigcup_{a\in V(T_r)} \chi(a) \setminus V_0
\]
for $1\leq r \leq m$.
Then the nonempty sets among $V_0,V_1,\dots,V_m$ form a partition of the set~$V(G)$.
Furthermore, there are no edges between any~$V_r$ and~$V_s$ with $1\leq r < s\leq m$.

Now set $Q_r = \{ j \in P_i \colon x_j \in V_r \}$ to be the set of processors heard from by~$i$ with values in~$V_r$ for $0\leq r \leq m$.
If $Q_0 \neq \emptyset$ then we are done:
In fact, we can choose any $k\in Q_0$ and $v = x_k$.
There are at least $\lvert P_i\rvert = n-f \geq f+1$ processors $j\in P_i$ whose all paths from~$x_j$ to~$x_k$ visit vertex $v = x_k$.

Suppose that $Q_0 = \emptyset$.
Note that at least two $Q_r$ are nonempty:
If~$Q_r$ with $r\geq 1$ is the only nonempty one, then $H_i^G \subseteq V_r$ and thus $H_i^T\subseteq T_r$, a contradiction to $b \in H_i^T$.
But then, for every~$r$ and every $j\in Q_r$ there exists some $s$ and $k\in Q_s$ with $r\neq s$.
Since~$V_0$ separates~$V_r$ and~$V_s$, all paths, in particular all chordless paths, from~$x_j$ to~$x_k$ visit some vertex $v\in V_0$.
By the pigeonhole principle, since there are at $\lvert P_i\rvert = n-f > \omega f $ such pairs $(j,k)$ and there are at most $\lvert V_0\rvert = \lvert \chi(b)\rvert \leq \omega$ vertices in~$V_0$, there exists one~$v$ for which there exist at least $f+1$ such pairs $(j,k)$.
\end{proof}

\begin{lemma}\label{lemma:chordal-validity}
Let $b \in H_i^T$ be an interior node of~$H_i^T$.
Then $\chi(b) \cap H_i^G \neq \emptyset$.
\end{lemma}
\begin{proof}
As in the proof of \lemmaref{lemma:intersecting-paths}, let $T_1, \dots, T_m$ be the branches of~$T$ when rooted at~$b$.
Define the sets~$V_r$ of vertices and~$Q_r$ of processors in the same fashion.
Fix some $v \in \chi(b)$ such that $\lvert \rho(v)\rvert \geq f+1$; by \lemmaref{lemma:intersecting-paths} such vertex $v \in V(G)$ exists.
To prove the lemma, we show that for any $J \subseteq {P_i \choose \lvert P_i\lvert - f }$ it holds that $v \in \langle \mathcal{V}_i^G(J) \rangle$.
Consider the set
\[
    \Psi = \{ j \in P_i : (x_j, \ldots, x_k) \in \rho(v) \} \subseteq P_i
\]
of processors such that $x_j \in V(G)$ starts a (possibly 0-length) path that intersects $v \in \chi(b)$.
Since $\lvert\rho(v)\rvert \geq f+1$ and $J = P_i \setminus A$ for some $A \subseteq P_i$ with $\lvert A\rvert=f$, there exists some $j^* \in \Psi \setminus A = \Psi \cap J$.
If $x_{j^*} = v$, then trivially $v \in \langle \mathcal{V}_i^G(J) \rangle$.
If $x_{j^*} \in V_r$ for some branch~$T_r$ with $r\geq 1$, then there exists some $k^* \in J \setminus Q_r$ since otherwise $b \notin \langle \mathcal{V}_i(J) \rangle_T$.

    Now $(j^*, k^*)\in Q_r\times Q_s$ and there exists a chordless path $(x_{j^*}=u_0, \ldots, v, \ldots, u_\ell = x_{k^*})$ that visits $v$. 
    By the definition of the monophonic convex hull, all vertices  $u_0, \ldots, u_\ell \in \langle \mathcal{V}_i^G(J) \rangle$ as they reside on a shortest path between $x_{j^*}, x_{k^*} \in \mathcal{V}_i^G(J)$.
\end{proof}

We are now ready to prove the main result on chordal graphs:

\begin{theorem}\label{thm:chordal-algorithm}
Let $G = (V,E)$ be a chordal graph.
If $n > (\omega(G)+1)f$, then monophonic approximate agreement on~$G$ can be solved in $O(\log \lvert V\rvert)$ asynchronous rounds.
\end{theorem}
\begin{proof}
  Let $X^G = X$ be the set of input values in $G$ and $X^T = \{b_i(0) : i \in P \setminus F \}$ be the set of initial bags of the expanded clique tree $(T,\chi)$ of $G$, where $x_i(0) \in \chi(b_i)$.
By \theoremref{thm:trees}, the approximate agreement algorithm on $T$ converges in $O(\log \lvert X^T \rvert)$ iterations.
By \lemmaref{lemma:expanded:upper}, we have $\log \lvert X^T \rvert = O(\log \lvert X^G\rvert)$.
The algorithm on~$T$ thus converges in $O(\log \lvert X^G \rvert)$ iterations.

The set of output vertices in~$T$ is $\{b,b'\} \subseteq E(T)$. 
Since $(T,\chi)$ is an expanded clique tree, we have that $\chi(b') \subseteq \chi(b)$ or vice versa, that is, $\chi(b) \cup \chi(b')$ is a clique.
Validity is satisfied as each non-faulty processor $i \in P \setminus F$ outputs a value $y_i \in \chi(b_i)$ such that $y_i \cap H_i^G$, which exists by \lemmaref{lemma:chordal-validity}.
Agreement is satisfied, as all output values reside in the same clique, and hence, the diameter of the output values is at most one. 
\end{proof}

Finally, we observe that the above implies that \emph{geodesic} approximate agreement can be solved in Ptolemaic graphs, as geodesic and monophonic convexities are identical on these graphs~\cite{Farber1986Convexity}.
\begin{corollary}
    If $G = (V,E)$ is a connected Ptolemaic graph and $n > (\omega(G)+1)f$, then geodesic approximate agreement on $G$ is solvable in $O(\log |V|)$ asynchronous rounds.
\end{corollary}

\section{Byzantine lattice agreement on cycle-free semilattices}\label{sec:lattice}

The abstract convex geometry framework can be applied to solve agreement problems on other combinatorial structures. As an example, we consider asynchronous Byzantine lattice agreement on a special class of semilattices. Let $\mathbb{L} = (V, \oplus)$ be a semilattice and $\le$ its natural partial order. The \emph{comparability graph} of $\le$ is the graph $G = (V,E)$ where $\{u,v\} \in E$ if $u\neq v$ and $u$ and $v$ are comparable. A partial order $\le$ is \emph{cycle-free} if the comparability graph is chordal~\cite{Ma1991cycle-free}. Similarly, we say $\mathbb{L}$ is cycle-free if $\le$ is cycle-free. See \figureref{fig:lattice-examples} for examples.

\begin{figure}
    \begin{center}
    \includegraphics[page=2,scale=0.95]{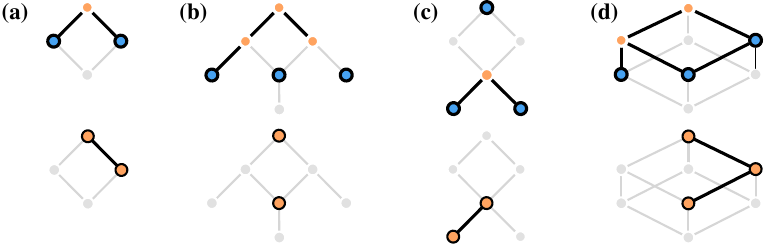}
    \end{center}
    \caption{Examples of algebraic convex sets on semilattices. The figures show the Hasse diagrams of these semilattices. In the top row, the blue vertices are input values and the orange vertices are contained in the hull of the blue vertices. The bottom row shows feasible outputs for these cases (i.e.\ chains contained in the convex hull). The semilattices (a)--(b) are cycle-free, whereas (c) and (d) respectively contain an induced 4- and 6-cycle in the comparability graph. \label{fig:lattice-examples}}
\end{figure}

\begin{lemma}\label{lemma:cycle-free-order}
    Let $\mathbb{L} = (V,\bigoplus)$ be a cycle-free semilattice. There exists an elimination order~$\preceq$ on the algebraic convexity of $\mathbb{L}$ such that $A(u) = \{ v : u \oplus v \in \{u,v\}, u \preceq v \}$ is a chain for any~$u \in V$.
\end{lemma}
\begin{proof}
    Let $\preceq$ be the perfect elimination order of the chordal comparability graph $G$ of $\mathbb{L}$. By definition of $\preceq$, for any $u \in V$ the set
    $A(u) = \{ \{u,v \} \in E : u \preceq v \}$
    is a clique, i.e., all elements in $A(u)$ are pairwise comparable in $\mathbb{L}$.
    We argue that $\preceq$ is also a convex elimination order for the algebraic convexity of $\mathbb{L}$.
    Let $V = \{v_1,\dots,v_N\}$ such that $v_i\preceq v_{i+1}$ for all $1\leq i < N$.
    We show by induction on~$i$ that $V_{i} = V \setminus \{ v_1, \ldots, v_{i} \}$ is convex. The base case $V_0=V$ is trivial. For the inductive step suppose $V_i$ is convex and consider the value $v_{i+1}$. Since $v_{i+1}$ is simplicial in $G[V_i]$, all the values $u \in V_i$ that are comparable with $v_{i+1}$ form a chain and $v_{i+1}$ is co-irreducible. Thus, $V_i \setminus \{ v_{i+1} \}$ is a subsemilattice of $\mathbb{L}$ and $v_{i+1} \in \ex V_{i+1}$.
\end{proof}

From now on  we let $\preceq$ denote the ordering given by \lemmaref{lemma:cycle-free-order} and define the output map
\[
    \phi(K) = \begin{cases} 
        \bigoplus K & \textrm{if } K \neq \ex K \\
        \max K & \textrm{otherwise.}
    \end{cases}
\]
With this, the framework given in \sectionref{sec:helly} and \lemmaref{lemma:without-min-progress} yield the following result.

\begin{lemma}\label{lemma:lattice-validity}
If $n > \max\{(\omega+1)f,3f\}$, then for all $i \in P \setminus F$ there exists $j \in P \setminus F$ such that 
    \[
    x_j \le y_i \le \bigoplus X.
    \]
\end{lemma}
\begin{proof}
    Let $i \in P \setminus F$. By \lemmaref{lemma:validity-combined} we have that $H_i \subseteq \langle X \rangle$. As $H_i$ and $\langle X \rangle$ are convex, i.e., closed under the $\oplus$ operator, we have $y_i = \phi(H_i) \in H_i \subseteq \langle X \rangle$ and $y_i \le \bigoplus X$.
    Moreover, for all $u \in \langle X \rangle$ there exists some $x_j \in X$ such that $x_j \le u $. In particular, this holds for any $y_i \in Y \subseteq \langle X \rangle$.
\end{proof}

\begin{lemma}\label{lemma:min-chain}
    If $\phi(H_i) = \min X$ for some $i \in P \setminus F$, then $Y$ is a chain.
\end{lemma}
\begin{proof}
    Suppose $y_i = \phi(H_i) = \min X$. Recall that $\min X \in \ex X$ by \remarkref{remark:min-is-ex}. First, we establish that $y_i = \min X$ implies that $H_i$ is a chain. If the second case of $\phi$ is used, then $H_i$ is trivially a chain. In the first case, we have that $y_i = \bigoplus H_i$, which is by definition comparable with each $u \in H_i$ and $\min X = \min \langle X \rangle \preceq \min H_i$, since $H_i \subseteq \langle X \rangle$. Thus, \lemmaref{lemma:cycle-free-order} implies that $H_i$ is a chain. 
    
    Since $H_i$ is a chain, we have $y_i = \phi(H_i) = \max H_i = \min X$, which yields that $H_i = \{ y_i \}$. Now consider any $y_j \in Y$. Since $y_j = \phi(H_j)$ and $H_j \cap H_i \neq \emptyset$ by \lemmaref{lemma:validity-combined}. Thus, $y_i \in H_j$ and $y_j$ is comparable with $y_i$. Since $Y \subseteq \langle X \rangle$ and $y_i \preceq y_j$ and \lemmaref{lemma:cycle-free-order} implies that $Y$ is a chain.
\end{proof}

\begin{theorem}\label{thm:lattice}
    Suppose $\mathbb{L} = (V, \oplus)$ is a cycle-free semilattice of height $\omega$ and $n > (\omega+1)f$. Then Byzantine semilattice agreement on $\mathbb{L}$ can be solved in the asynchronous model.
\end{theorem}
\begin{proof}
  Validity is maintained by \lemmaref{lemma:lattice-validity}. By \lemmaref{lemma:without-min-progress}, if $\min X \notin Y$ holds, then $\langle Y \rangle \subsetneq \langle X \rangle$. Since $Y$ is always nonempty, eventually $\min X \in Y$. Then by \lemmaref{lemma:min-chain} the set $Y$ is a chain. As $\langle X \rangle$ can shrink at most $|V|$ times, the processors can decide and terminate after $|V|$ iterations.
\end{proof}

\section{Synchronous convex consensus\label{apx:synchronous}}

In this section, we consider the following \emph{exact} agreement problem over an arbitrary convexity space $\mathcal{C}$ in the synchronous model of computation. In the \emph{convex consensus problem on a convexity space $\mathcal{C}$}, the task is to satisfy the following constraints:
\begin{itemize}[noitemsep]
    \item agreement: $|Y| = 1$ (all non-faulty processors decide on a single value),
    \item validity: $Y \subseteq \langle X \rangle$ (the decided value resides in the convex hull of the input values).
\end{itemize}

We now establish that the convex consensus on any convex geometry is solvable in $\Theta(f)$ synchronous communication rounds if and only if $n > \max\{ 3f, \omega f\}$ is satisfied. In particular, the upper bound holds for \emph{any} arbitrary convexity space. These results can be seen as generalisation of the results for $\mathbb{R}^m$ given by Mendes et al.~\cite{vaidya13byzantinevector,MHVG15} to arbitrary, discrete convexity spaces.

\begin{restatable}{theorem}{syncupper}\label{thm:synchronous-upper}
    Let $\mathcal{C}$ be an abstract convexity space on $V$ with Helly number $\omega$. If $n > \max \{3f, \omega f \}$, then convex consensus on $\mathcal{C}$ can be solved in $O(f)$ synchronous communication rounds using $O( n f^2)$ messages of size $O(n \cdot (\log n + \log |V|))$.
\end{restatable}

The next result establishes that the higher resilience threshold of $n > \omega f$ for convex consensus is necessary already in the case of convex geometries.

\begin{restatable}{theorem}{synclower}\label{thm:synchronous-lower}
    Let $\mathcal{C}$ be a convex geometry with a Helly number $\omega$. If $n \le \omega f$ holds, then convex consensus on $\mathcal{C}$ cannot be solved in the synchronous message-passing model.
\end{restatable}

\subsection{Synchronous model of computation}

Recall that in the standard synchronous message-passing model, where the computation proceeds in discrete rounds, where each non-faulty processor performs in lock-step the following:
\begin{enumerate}[noitemsep]
    \item send messages to other processors in the system,
    \item receive messages from all other processors (or no message from a faulty processor),
    \item update local state based on received messages.
\end{enumerate}

\subsection{Upper bound for convex consensus}

We start with the positive result given by \theoremref{thm:synchronous-upper}. The algorithm follows the same idea as the Vaidya--Garg algorithm for Euclidean spaces~\cite{vaidya13byzantinevector,MHVG15}. For the sake of completeness, we reiterate the algorithm here and generalise the analysis for any finite abstract convexity space.

\syncupper

Note that in the above the convexity space $\mathcal{C}$ need not be a \emph{convex geometry}; indeed, it suffices that the Helly number $\omega$ is finite. Moreover, observe that standard multivalued consensus protocols do not directly solve the problem, as they have weaker validity constraints (e.g., validity requires that if all agree on a single value, then this must be decided on).

\subsubsection{Multivalued Byzantine agreement and broadcast} As a subroutine, we use the following variant of safe broadcast in the synchronous model. Let $M$ be a finite set of messages and $s \in P$ be a fixed sender with some input message $m \in M$. The task is to have all correct processors $i \in P \setminus F$ decide on a value $y_i \in M \cup \{ \bot \}$ to the following constraints:
\begin{itemize}[noitemsep]
    \item agreement: $y_i \in M \cup \{ \bot \}$ and $|Y| = 1$ (all correct processors agree on a single message)
    \item validity: if $s \notin F$, then $Y = \{ m \}$. 
\end{itemize}
In words, all correct processors agree on a single message or $\bot$ denoting that the sender is faulty. However, if the sender $s$ is correct, then all correct nodes must decide on the message $m$.

\begin{theorem}[Theorem 4 in \cite{srikanth1987simulating}]
    Let $M$ be a finite set. If $n > 3f$, then multivalued Byzantine agreement on $M$ can be solved in $2f+2$ synchronous rounds using $O(nf^2)$ messages of size at most $O(\log n + \log |M|)$.
\end{theorem}

\subsubsection{From multivalued Byzantine agreement to convex consensus}

With the above subroutine, we can see that convex consensus reduces to the problem of solving multivalued Byzantine agreement with linear-in-$n$ overhead to the bit complexity.

\begin{theorem}
    Let $\mathcal{C}$ be an abstract convexity space on $V$ with Helly number $\omega$. Consider a network of $n$ nodes with $f$ faulty nodes. Suppose there exists an algorithm $\vec A$ that solves multivalued Byzantine agreement on $V$ in a network with $n$ nodes and $f$ Byzantine faulty nodes. 
     If $n > \omega f$, then convex consensus on $\mathcal{C}$ can be solved in the same time as multivalued Byzantine agreement.
\end{theorem}
\begin{proof}
 Let $x_i \in V$ be the input for processor $i \in P$. Suppose we run $n$ parallel copies, $\vec A_1, \ldots, \vec A_n$ of a multivalued Byzantine agreement algorithm such that for instance $\vec A_i$ processor $i$ acts as the sender with the message $m_i = x_i \in V$. After the algorithms $\vec A_1, \ldots, \vec A_n$ have terminated, every non-faulty processor $i \in P \setminus F$ has decided for all $j \in P$ on the message $m_j$ transmitted by processor $j$. In case $j \in F$ and $m_j = \bot$, we can arbitrarily map $m_j$ to some default value in $V$ without loss of generality.

Unlike in the asynchronous case, all correct processors have identical views after the algorithms $\vec A_1, \ldots, \vec A_n$ terminate. After this point, no further communication is needed. It only remains to compute an output value $y \in \langle X \rangle$ consistently. For $J \subseteq P$, let $\mathcal{V}(J) = \{ m_j : j \in J\}$ denote the set of messages from processors in $J$. Each processor locally computes
 \[
 \mathcal{K} = \left\{ \langle \mathcal{V} (J) \rangle : J \in {P \choose n-f} \right\} \quad \textrm{ and } \quad H = \bigcap \mathcal{K}.
\]
As before, it is easy to check that $H \subseteq \langle X \rangle$. To see that $H \neq \emptyset$, we show that $\mathcal{K}$ is $\omega$-intersecting by picking $A_1, \ldots, A_\omega \in \mathcal{K}$. Since $A_k = \langle \mathcal{V}(J_k) \rangle$ for some $J_k \in  {P \choose n-f}$, applying \lemmaref{lemma:intersection-size} it follows that 
\[
\left |\bigcap_{k=1}^{\omega} J_k \right| \ge |P| - \sum_{k=1}^{\omega} \left|P \setminus J_k \right| \ge n - \omega f > 0.
\]
Thus, $\mathcal{K}$ is $\omega$-intersecting and the convexity space $\mathcal{C}$ has Helly number $\omega$, so we have $H = \bigcap \mathcal{K} \neq \emptyset$. Given that all correct processors use the same output function $\phi$, they decide on the same output value $\phi(H) \in H \subseteq \langle X \rangle$. This satisfies both agreement and validity.
\end{proof}

\subsection{Lower bounds for convex consensus}

Clearly, convex consensus on any non-trivial convexity space $\mathcal{C}$ is at least as hard as binary consensus. Therefore, classic lower bounds for synchronous binary consensus immediately imply that solving convex consensus requires $\Omega(f)$ synchronous rounds and a resilience condition of $n > 3f$. The following observation shows that the bound $n > \omega f$ on resilience is necessary for convex consensus already when $\mathcal{C}$ is a convex geometry.

\synclower    
\begin{proof}
    Suppose algorithm $\vec A$ solves convex consensus on $\mathcal{C}$. Since $\mathcal{C}$ is a convex geometry with Helly number $\omega$, it has a free set $A \in \mathcal{C}$ of size $|A|=\omega$. Let $C_1, \ldots, C_k$ be a partition of the $n = kf$ processors into $k$ disjoint sets and define the inputs as $x_i = a_j$ for $i \in C_j$. Let $\Xi_1$ be the execution of~$\vec A$ when $F = \emptyset$. By agreement and validity, all processors decide on some value $a_j \in A$. Now consider an execution $\Xi_2$ of $\vec A$, where $F = C_j$ is the set of Byzantine processors that behave exactly as $C_j$ in the execution $\Xi_1$. For any $i \in P \setminus C_j$ the two executions $\Xi_1$ and $\Xi_2$ are indistinguishable, so the correct processors output $a_j \in A$. However, this violates the validity constraint, as $a_j \notin \langle X \rangle = A \setminus \{ a_j \}$.
\end{proof}

\section{Resilience lower bounds for abstract convexity spaces}\label{sec:lower-bounds}

We establish general lower bounds for asynchronous approximate agreement on abstract convexity spaces using a partitioning argument and so-called blocking instances, which are given by irredundant and free sets of the underlying convexity space. This approach can be seen as a generalisation of the Mendes et al.~\cite{MHVG15} lower bound technique from Euclidean spaces into arbitrary convexity spaces. The Carath\'eodory number of $\mathcal{C}$ equals the maximum cardinality of an irredundant set in $\mathcal{C}$.

\begin{theorem}\label{thm:general-lb}
    Let $\mathcal{C}$ be a convexity space with Carath\'eodory number $c$ and Helly number $\omega$. Then:
    \begin{itemize}
        \item If $n \le (c+1)f$, then there is no asynchronous abstract approximate agreement algorithm on $\mathcal{C}$ that on all inputs satisfies validity and agreement (i.e., the set of outputs is free).
        \item If $\mathcal{C}$ is a convex geometry and $n \le (\omega+1)f$, then there is no asynchronous abstract approximate agreement algorithm on $\mathcal{C}$ satisfying validity that outputs at most $\omega-1$ distinct values.
    \end{itemize}
\end{theorem}

Combining the above result with classic results in combinatorial convexity theory gives lower bounds for specific problems. For any graph $G$ with diameter at least two, the Carath\'eodory number is two~\cite{Duchet1988Convex} and clique is a free set. This implies the following result. 

\begin{corollary}
    The monophonic approximate agreement problem on any $G$ with diameter at least two cannot be solved if $n \le 3f$. There is no asynchronous algorithm that outputs a clique of size less than $\omega$ unless $n \le (\omega+1)f$.
\end{corollary}

The case of Byzantine semilattice agreement is perhaps more interesting, as the breadth of a semilattice coincides with its Carath\'eodory number~\cite{jamison1982perspective}. For any $b>1$, there are semilattices with height and breadth equal to $b$: take the subsemilattice of a subset lattice over $[b]$ without $\emptyset$. 

\begin{corollary}
    Suppose $\mathbb{L}$ is a semilattice with breadth $b$. If $n \le (b+1)f$, then there exists no asynchronous algorithm that solves Byzantine semilattice agreement on $\mathbb{L}$. 
\end{corollary}

\subsection{Partitioning and blocking instances}

\begin{definition}\label{definition:blocking}
    Let $\mathcal{C}$ be an abstract convexity space on $V$, $A \subseteq V$, and $\mu \colon A \times \langle A \rangle \to A$. We say that $(A,\mu)$ is an $m$-\emph{blocking instance} for $\mathcal{C}$ if the following conditions are satisfied:
    \begin{enumerate}[noitemsep]
        \item $|A| = m$,
        \item $A = \ex A$, 
        \item $\mu(x,y) \neq x$ for all $x \neq y$
        \item $y \notin \langle A \setminus \mu(x,y) \rangle$ for all $x \neq y$.
    \end{enumerate}
\end{definition}

\begin{theorem}\label{thm:blocking}
    Suppose there exists an $m$-blocking instance $(A,\mu)$ for $\mathcal{C}$. If $n \le (m+1)f$, then there does not exist an $f$-resilient algorithm which outputs $Y \subsetneq \langle X \rangle$ for all input sets $X \in {V \choose m }$.
\end{theorem}
\begin{proof}
    For the sake of contradiction, suppose there exists an algorithm $\mathbf{A}$ that outputs $Y \subsetneq \langle X \rangle$ for any $X \in {V \choose m }$. We show that the algorithm fails when $X = A$. Without loss of generality, assume $n = (m+1)f$. Let $\{ C_1, \ldots, C_m, B\} \subseteq {P \choose f}$ be partition of the processor identifiers into disjoint sets of size $f$ and $A = \{ a_1, \ldots, a_m \}$. We define $\gamma \colon P \times [m] \to [m]$ as
\[
 \gamma(i,z) = \begin{cases}
     a_k & \textrm{ if } i \in C_k \\
     z   & \textrm{ otherwise}.
 \end{cases}
\]

    First, consider the scenario where $F=B$ and the inputs are given by $x_i = \gamma(i,z)$ so that $X=A$. Let $\xi$ be an admissible schedule for the execution $\Xi$, where all the faulty processors immediately crash before sending any messages. Since by assumption the output satisfies $Y \subsetneq \langle X \rangle = \langle A \rangle$ and $A = \ex A$ holds by condition 2 of \definitionref{definition:blocking}, there must be some non-faulty processor $j \in C_k$ that outputs $x_j \neq y_j$. Note that since no processor in $B$ sends a message, $\xi$ is an admissible schedule for any scenario with inputs $x_i = \gamma(i,z)$ independent of the choice of $z \in A$. 

    Fix $a_h = \mu(x_j, y_j)$. Note that as $x_j \in C_k$, we have $x_j = a_k$ and $h \neq k$. Consider the following execution $\Xi'$ of $\mathbf{A}$, where
    \begin{itemize}
     \item the set of faulty processors is $F = C_h$,
     \item the inputs are given by $x_i = \gamma(i, a_k)$, and thus, $X = A \setminus a_h$,
     \item the set $B$ of correct processors is indefinitely delayed,
     \item the faulty processors behave exactly as $C_h$ in the execution $\Xi$, and
     \item the schedule $\xi'$ has $\xi$ as its prefix.
    \end{itemize}
    Clearly, $\xi$ is an admissible schedule prefix for $\Xi'$ and the two executions $\Xi$ and $\Xi'$ are it is indistinguishable for any $j \in P \setminus F$ at least until $j$ decides on its output value $y_j = a_h$. Let $\xi'$ be an extension of $\xi$, where all the Byzantine processors in $F=C_h$ crash or change their input values arbitrarily immediately after $j$ has decided on its output. The fact that $j \in P \setminus F$ decides on value $a_h \notin \langle X \rangle = \langle A \setminus a_h \rangle$, contradicts the assumption that $\mathbf{A}$ maintains validity.
\end{proof}
\subsection{Blocking instances from irredundant and free sets}

\begin{remark}
    If $A$ is irredundant, then $A = \ex A$ holds.
\end{remark}
\begin{proof}
    Suppose $A$ is irredundant and there is some $a \in A \setminus \ex A$. Since $a \notin \ex A$, we have $a \in \langle A \setminus a \rangle = \langle A \rangle$ and $\partial A = \emptyset$. Thus, $A$ is not irredundant.
\end{proof}

\begin{lemma}
    Let $\mathcal{C}$ be a convexity space on $V$ and $A \subseteq V$ such that $|A| > 1$. Then for any $a \in A$ and $y \in \langle A \rangle \setminus A$ there exists some $b(a,y) \in A \setminus a$ such that $y \notin \langle A \setminus b(a,y) \rangle$.\label{lemma:bad-points}
\end{lemma}
\begin{proof}
    Note that the claim is vacuous if $A$ is redundant. Consider the sets 
    \[
    B = \bigcup_{b \in A } \langle A \setminus b \rangle \quad \textrm{ and } \quad \partial A = \langle A \rangle \setminus B.
    \]
    Since $|A| > 1$, for any $c \in A$ there is some $c' \in A \setminus c$ such that $c \in A \setminus c' \subseteq \langle A \setminus c' \rangle$. Thus, 
    \[
    A \subseteq \bigcup_{b \in A \setminus a} \langle A \setminus b \rangle \subseteq \bigcup_{b \in A } \langle A \setminus b \rangle = B.
    \]
    Thus, we have that the inclusions $A \subseteq B$ and $\partial A \subseteq \langle A \rangle \setminus A$ hold.

    For the sake of contradiction, fix some $a \in A$ and suppose some $y \in \langle A \rangle \setminus A$ that violates the claim of the lemma, that is, $y \in \langle A \setminus b \rangle$ holds for all $b \in A \setminus a$. This implies that 
    \[
    y \in \bigcup_{b \in A \setminus a } \langle A \setminus b \rangle \subseteq B.
    \]
    However, $y \in B$ implies that 
    \[
    y \notin \partial A = \langle A \rangle \setminus B \subseteq \langle A \rangle \setminus A,
    \] 
    which contradicts the assumption that $y \in \langle A \rangle \setminus A$.
\end{proof}

\begin{lemma}
    If $A$ is irredundant and $|A|=m > 1$, then there exists an $m$-blocking instance $(A,\mu)$.
\end{lemma}
\begin{proof}
    If $A$ is irredundant, we have $\partial A \neq \emptyset$ and $A = \ex A$. Moreover, if $|A|=m$, then $A$ satisfies the first two conditions of \definitionref{definition:blocking}. To obtain $\mu$, we define for all $x \in A$ and $y \in \langle A \rangle$
    \[
    \mu(x, y) = \begin{cases}
        y & \textrm{if } y  \in A, \\
        b(x,y) & \textrm{otherwise},
    \end{cases}
    \]
    where $b(x,y)$ is the value given by \lemmaref{lemma:bad-points}. It remains to check that $\mu$ satisfies the remaining two conditions given in \definitionref{definition:blocking}. Suppose $x \neq y$. If $y \in A$, then $\mu(x,y) = y \neq x$ and since $y \in A = \ex A$, we also have $y \notin \langle A \setminus y \rangle$. Otherwise, $y \in \langle A \rangle \setminus A$ and \lemmaref{lemma:bad-points} yields that the value $b(x,y) \in A \setminus x$ satisfies both conditions for $\mu$.
\end{proof}

\begin{lemma}
    If $A$ is free and $|A|=m>1$, then there exists an $m$-blocking instance $(A,\mu)$.
\end{lemma}
\begin{proof}
    Since $A$ is free it satisfies $A = \langle A \rangle = \ex A$. Therefore, it suffices to define $\mu : A \times A \to A$ as $\mu(x,y) = y$, which clearly satisfies the conditions given in \definitionref{definition:blocking}.
\end{proof}

\theoremref{thm:general-lb} now follows by recalling that the (1) Carath\'eodory number $c$ of an abstract convexity space $\mathcal{C}$ equals the maximum cardinality of an irredundant set, and that (2) for convex geometries, the Helly number equals the maximum cardinality of a free set~\cite{Edelman1985}.

\section{Conclusions}
\label{sec:conclusions}

Many structured agreement tasks correspond to exact or approximate agreement problems on (possibly discrete) convexity spaces. Using the theory of abstract convexity, we have obtained Byzantine-tolerant algorithms for a large class of agreement problems on discrete combinatorial structures. 
In the synchronous model, exact convex consensus for any convexity space can be solved in an optimally resilient manner with asymptotically optimal round complexity. However, in the asynchronous setting, several interesting open problems remain. 

\begin{enumerate}
    \item It seems difficult to come up with a general rule for the output map $\phi : \mathcal{C} \to V$ in a way that guarantees that the convex hull of active values shrinks. Nevertheless, we have seen that on chordal graphs and cycle-free semilattices we can solve approximate agreement efficiently. In both cases, the underlying convexity space is a convex geometry.  Given that the literature is abound with convex geometries associated with combinatorial structures~\cite{Jamison1984Helly,Farber1986Convexity,Farber1987Local,Duchet1988Convex,Pelayo2013,nielsen2009steiner,dourado2013caratheodory,duchet1987convexity,Edelman1985,poncet2014semilattice}, it is natural to ask whether the abstract approximate agreement problem can be solved for other convex geometries as well. 

\item It is unclear whether the abstract approximate agreement problem can be solved on general convexity spaces. For example, the asynchronous algorithms for approximate agreement on graphs presented here fail for non-chordal graphs: already the simplest example of a non-chordal graph, the four cycle, is difficult to handle. Indeed, the monophonic convexity of a four cycle is \emph{not} a convex geometry: a convex set may not necessarily have any extreme points, and thus, greedily excluding extreme points does not seem to work. Are there resilient asynchronous algorithms that solve the problem for non-chordal graphs?

\item We obtained resilience lower bounds in terms of the Carath\'eodory and the Helly numbers. However, our positive results for the asynchronous model hold in cases where the Carath\'eodory number is at most two. Interestingly, in the continuous setting of multidimensional approximate agreement~\cite{MHVG15}, tight resilience bounds exist, as the Carath\'eodory and Helly numbers coincide in the usual Euclidean convexity space~on~$\mathbb{R}^m$. Is there a discrete convexity space with a higher Carath\'eodory number in which approximate agreement can be solved?
\end{enumerate}

\section*{Acknowledgements}

We thank the anonymous reviewers for their helpful comments and Janne H. Korhonen for many discussions regarding this work. We also wish to thank the participants of the Helsinki Workshop on Theory of Distributed Computing 2018 and the Metastability workshop in Mainz 2018 for discussions that lead to the problem of approximate agreement on graphs. This work was supported by the CNRS project PEPS DEMO and the Universit\'e Paris-Saclay project DEPEC MODE (T.N.). This project has received funding from the European Union's Horizon 2020 research and innovation programme under the Marie Sk\l{}odowska-Curie grant agreement No 754411 (J.R.).

\bibliography{convexity}

\end{document}